\documentclass[12p,final]{elsarticle}
\usepackage[margin=2.5cm]{geometry}% by courtesy of Mico
\usepackage{lipsum}

\makeatletter
\def\ps@pprintTitle{%
   \let\@oddhead\@empty
   \let\@evenhead\@empty
   \def\@oddfoot{\reset@font\hfil\thepage\hfil}
   \let\@evenfoot\@oddfoot
}
\makeatother
\usepackage{amsfonts,color,morefloats,pslatex}
\usepackage{amssymb,amsthm, amsmath,latexsym}

\newtheorem{theorem}{Theorem}
\newtheorem{lemma}[theorem]{Lemma}

\newcommand{\tr}{{\mathrm{Tr}}}

\newcommand{\Norm}{{\mathrm{N}}}

\newcommand{\gf}{{\mathbb{F}}}
\newcommand{\wt}{{\mathtt{wt}}}

\newcommand{\C}{{\mathcal{C}}}

\newcommand{\bc}{{\mathbf{c}}}

\usepackage{blindtext}

\begin{document}
%\tableofcontents

\begin{frontmatter}

%% Title, authors and addresses

%% use the tnoteref command within \title for footnotes;
%% use the tnotetext command for the associated footnote;
%% use the fnref command within \author or \address for footnotes;
%% use the fntext command for the associated footnote;
%% use the corref command within \author for corresponding author footnotes;
%% use the cortext command for the associated footnote;
%% use the ead command for the email address,
%% and the form \ead[url] for the home page:
%%
%% \title{Title\tnoteref{label1}}
%% \tnotetext[label1]{}
%% \author{Name\corref{cor1}\fnref{label2}}
%% \ead{email address}
%% \ead[url]{home page}
%% \fntext[label2]{}
%% \cortext[cor1]{}
%% \address{Address\fnref{label3}}
%% \fntext[label3]{}

\title{A family of projective two-weight linear codes}

\tnotetext[fn1]{
The research of this paper was supported in part by the National Natural Science Foundation of China under Grant 11901049 and 11971004, in part by the the Natural Science Basic Research Program of Shaanxi under Grant 2020JQ-343, and in part by the  Young Talent Fund of University Association for Science and Technology in Shaanxi, China, under Grant 20200505.
}

\author[heng]{Ziling Heng}
\ead{zilingheng@chd.edu.cn}
\author[heng]{Dexiang Li}
\ead{ lidexiang1015@163.com}
\author[du1,du2]{Jiao Du}
\ead{jiaodudj@126.com}
\author[heng]{Fuling Chen}
\ead{chenfuling1109@163.com}

%\cortext[zlheng]{Corresponding author}
\address[heng]{School of Science, Chang'an University, Xi'an 710064, China}
\address[du1]{School of Mathematics and Information Science, Henan Normal University, Xinxiang 453007, China}
\address[du2]{Henan Engineering Laboratory for Big Data Statistical Analysis and Optimal Control, Henan Normal University, Xinxiang 453007, China}

%% use optional labels to link authors explicitly to addresses:
%% \author[label1,label2]{<author name>}
%% \address[label1]{<address>}
%% \address[label2]{<address>}
%\author{Cunsheng Ding}
%\ead{cding@ust.hk}

%\cortext[lcj]{Corresponding author}
%\address{Department of Computer Science and Engineering,
%The Hong Kong University of Science and Technology,
%Clear Water Bay, Kowloon, Hong Kong, China}

%\tableofcontents

\begin{abstract}
Projective two-weight linear codes are closely related to finite projective spaces and strongly regular graphs. In this paper, a family of $q$-ary projective two-weight linear codes is presented, where $q$ is a power of 2. The parameters of both the codes and their duals are excellent. As applications, the codes are used to derive strongly regular graphs with new parameters and secret sharing schemes with interesting access structures.
\end{abstract}

\begin{keyword}
Linear code \sep strongly regular graph \sep secret sharing schemes
%% PACS codes here, in the form: \PACS code \sep code

%% MSC codes here, in the form: \MSC code \sep code
%% or \MSC[2008] code \sep code (2000 is the default)
\MSC  94B05 \sep 05E30 \sep 15A03

\end{keyword}

\end{frontmatter}

%\tableofcontents

\section{Introduction}
An $[n,k,d]$ \emph{linear code} $\C$ of length $n$ over $\gf_q$ is a $k$-dimensional linear subspace of $\gf_q^n$ with minimal Hamming distance $d$, where $\gf_q$ denotes the finite field with $q$ elements for a prime power $q$. There exist some bounds on the parameters $n,k,d$ \cite{MS77}. An $[n,k,d]$ code is called (distance) optimal if no $[n,k,d+1]$ codes exists. An $[n,k,d-1]$ code is called almost (distance) optimal if $[n,k,d]$ code is optimal. The \emph{dual} code $\C^{\perp}$ of a linear code $\C$ is defined as $\C^{\perp}=\left\{\textbf{c}^{\perp}\in \gf_q^n: \langle\textbf{c}^{\perp},\textbf{c}\rangle=0\ \forall\ \textbf{c}\in \C\right\}$. Clearly, the dimension of $\C^{\perp}$ is $n-k$ if $\C$ has dimension $k$. A linear code $\C$ is said to be \emph{projective} if the minimal distance of its dual is at least 3.

Let $A_i$ denote the number of codewords with weight $i$ in a code $\C$ of length $n$. The sequence $(1,A_1,A_2,\cdots,A_n)$ is called the \emph{weight distribution} of $\C$ and the polynomial $1+A_1z+A_2z^2+\cdots+A_nz^n$ is called the \emph{weight enumerator} of $\C$.  Denote by $t=\sharp \{A_i:A_i>0,\ 1\leq i\leq n\}$. Then a code with weight distribution $(1,A_1,A_2,\cdots,A_n)$ is called a $t$-weight code. Weight distribution is an interesting research subject in coding theory as it not only contains the information of the capabilities of error detection and correction, but also allows the computation of the error probability of error detection and correction of a given code. Many linear codes with only a few weights were reported in the literature \cite{C, D15, DD, D72, HDZ, HWW, SA, TLQZ, TQH, YZ, YD, ZB, Z}. In particular, projective two-weight codes are very precious as they are closely related to finite projective spaces, strongly regular graphs and  combinatorial designs \cite{C, Dingbook18, D72}. However, projective two-weight codes are rare and only a few families are known \cite{C, D15, HY}.

Let $D=\{d_1,d_2,\cdots,d_n\}\subseteq \gf_{q^m}$ be a nonempty set. Let $\tr_{q^m/q}$ be the trace function from $\gf_{q^m}$ to $\gf_q$. Define a linear code of length $n$ as
\begin{eqnarray}\label{defn-code}\C_D=\{(\tr_{q^m/q}(bd_1), \tr_{q^m/q}(bd_2),\cdots, \tr_{q^m/q}(bd_n)):b\in \gf_{q^m}\}.\end{eqnarray}
The set $D$ is called the \emph{defining set}  and the ordering of the elements in $D$ doesn't affect the parameters and weight distribution of $\C_D$. Recently, many good codes were obtained with different proper defining sets \cite{D15, DD, HWW, HY, LB, TLQZ, TQH, YZ, ZB, Z}.

Let $s$ be a divisor of a positive integer $m$, $\Norm_{q^m/q^{s}}$ the norm function from $\gf_{q^m}^*$ to $\gf_{q^{s}}^*$. In this paper, we study the $q$-ary linear code $\C_D$ in (\ref{defn-code}) with the defining set
$$D=\left\{x\in \gf_{q^m}^*:\tr_{q^{s}/q}(\Norm_{q^m/q^{s}}(x))+c=0\right\},\ c\in \gf_q.$$
Known results on the code $\C_D$ are summarized as follows:
\begin{enumerate}
\item If $q$ is an odd prime, $c=0$ and $m=2s$, the code $\C_D$ is a two-weight code and its complete weight enumerator was given in \cite{LB}.
\item If $c=0$ and $s=2$, the code $\C_D$ is a two-weight code and its weight enumerator was given in \cite{HY}.
\item If $c\in \gf_q^*,s=2$ and $\gcd(m/2,q-1)=1$, the code $\C_D$ is a two-weight code and its weight enumerator was given in \cite{HY}.
\end{enumerate}
The objective of this paper is to investigate the weight distribution of $\C_D$ for $m=2s$, any prime power $q$ and any $c\in \gf_q$. The parameters of its dual are also determined. It is interesting that $\C_D$  is a projective two-weight code if $m=2s$, $q$ is even and $c\in \gf_q^*$. As an application, strongly regular graphs with new parameters can be derived with the projective two-weight codes. As another application, all codes obtained in this paper can be used to construct secret sharing schemes with interesting access schemes.

\section{Preliminaries}
Let $p$ be a prime and $q=p^e$. Denote by $\zeta_p$ the primitive $p$-th root of complex unity. For any $a\in \gf_q$, the function
$\varphi_{a}(x)=\zeta_{p}^{\tr_{q/p}(ax)},\ x\in \gf_q$, is called  an \emph{additive character} of $\gf_q$. If $a=0$, $\varphi_0(x)=1$ for all $x\in \gf_q$ and $\varphi_0$ is called the trivial additive character of $\gf_q$. If $a=1$, $\varphi_1$ is said to be the canonical additive character of $\gf_q$. The orthogonality  relation of additive characters (see \cite{LN}) is given by
$$\sum_{x\in \gf_q}\varphi_1(ax)=\left\{
\begin{array}{rl}
q    &   \mbox{ for }a=0,\\
0    &   \mbox{ for }a\in \gf_q^*.
\end{array} \right. $$
Let $\gf_q^*=\gf_q\setminus \{0\}$ and $\alpha$ be a primitive element of it. Define the \emph{multiplicative characters} of $\gf_q$ by
$\psi_{j}(\alpha^k)=\zeta_{q-1}^{jk}\mbox{ for }k=0,1,\cdots,q-1$, where $0\leq j \leq q-2$. In particular, $\psi_0$ is called the trivial multiplicative character and $\eta:=\psi_{(q-1)/2}$ is referred to as the quadratic multiplicative character of  $\gf_q$. The orthogonality relation of multiplicative characters (see \cite{LN})  is given by
$$\sum_{x\in \gf_q^*}\psi_j(x)=\left\{
\begin{array}{rl}
q-1    &   \mbox{ for }j=0,\\
0    &   \mbox{ for }j\neq 0.
\end{array} \right. $$

For an additive character $\varphi$ and a multiplicative character $\psi$ of $\gf_q$, the \emph{Gauss sum} $G(\psi, \varphi)$ over $\gf_q$ is defined by
$$G(\psi,\varphi)=\sum_{x\in \gf_q^*}\psi(x)\varphi(x).$$
We call $G(\eta,\varphi)$ the quadratic Gauss sum over $\gf_q$ for nontrivial $\varphi$.

\begin{lemma}\label{quadGuasssum}\cite[Th. 5.15]{LN}
Let $q=p^e$ with $p$ odd. Let $\chi$ be the canonical additive character of $\gf_q$. Then
\begin{eqnarray*}G(\eta,\varphi)&=&(-1)^{e-1}(\sqrt{-1})^{(\frac{p-1}{2})^2e}\sqrt{q}\\
 &=&\left\{
\begin{array}{lll}
(-1)^{e-1}\sqrt{q}    &   \mbox{ for }p\equiv 1\pmod{4},\\
(-1)^{e-1}(\sqrt{-1})^{e}\sqrt{q}    &   \mbox{ for }p\equiv 3\pmod{4}.
\end{array} \right. \end{eqnarray*}
\end{lemma}

\begin{lemma}\label{lem-charactersum}\cite[Th. 5.33]{LN}
Let $\varphi$ be a nontrivial additive character of $\gf_q$ with $q$ odd, and let $f(x)=a_2x^2+a_1x+a_0\in \gf_q[x]$ with $a_2\neq 0$. Then
$$\sum_{c\in \gf_q}\varphi(f(c))=\varphi(a_0-a_1^2(4a_2)^{-1})\eta(a_2)G(\eta,\varphi).$$
\end{lemma}

\begin{lemma}\label{lem-charactersum-evenq}\cite[Cor. 5.35]{LN}
Let $\varphi_b$ be a nontrivial additive character of $\gf_q$ with $b\in \gf_q^*$, and let $f(x)=a_2x^2+a_1x+a_0\in \gf_q[x]$ with $q$ even. Then
$$\sum_{c\in \gf(q)}\varphi_b(f(c))=\left\{\begin{array}{ll}
\chi_b(a_0)q    &   \mbox{ if }a_2=ba_{1}^{2},\\
0    &   \mbox{ otherwise. }
\end{array} \right.$$
\end{lemma}
\section{A class of $q$-ary two-weight linear codes}
Let $s$ be a positive integer and $m=2s$. Let $q$ be a power of a prime $p$. In this section, we study the code $\C_D$ in (\ref{defn-code}) with the defining set
$$D=\left\{x\in \gf_{q^m}^*:\tr_{q^s/q}(x^{q^s+1})+c=0 \right\}\ \text{ for }c\in \gf_q.$$
Denote by  $n=\sharp D$. The following lemma is useful for determining the length of $\C_{D}$.

\begin{lemma}\label{lem-sums}\cite[Lemma 7]{HWD}
For $a\in \gf_{q^s}^*, b \in \gf_{q^m},c\in\gf_{q}$ and $m=2s$,
define an exponential sum by
$$\Delta(a,b,c)=\sum_{y\in \gf_{q}^*}\varphi(yc)\sum_{x\in \gf_{q^m}^*}\varphi \left(y\left( \tr_{q^s/q}(ax^{q^s+1})+\tr_{q^m/q}(bx)\right)\right),$$
where $\varphi$ is the canonical additive character of $\gf_{q}$. Then
\begin{eqnarray*}
\Delta(a,b,c)=
\begin{cases}
(q^s+1)(1-q) & \text{ if $c=0,\tr_{q^s/q}\left(\frac{b^{q^s+1}}{a}\right)=0$},\\
q^s+1-q & \text{ if $c=0,\tr_{q^s/q}\left(\frac{b^{q^s+1}}{a}\right)\neq0$},\\
q^s+1 & \substack{\mbox{ if }c\neq0, \ \tr_{q^s/q}\left(\frac{b^{q^s+1}}{a}\right)\neq0, \ c\neq\tr_{q^s/q}\left(\frac{b^{q^s+1}}{a}\right), \mbox{ or } c\neq0, \ \tr_{q^s/q}\left(\frac{b^{q^s+1}}{a}\right)=0,}\\
q^s+1-q^{s+1} & \text{ if $c=\tr_{q^s/q}\left(\frac{b^{q^s+1}}{a}\right)\neq0$}.
\end{cases}
\end{eqnarray*}
\end{lemma}

\begin{lemma}\label{length}
The length $n$ of $\C_{D}$ is given by
\begin{eqnarray*}
n=\left\{\begin{array}{ll}
(q^s+1)(q^{s-1}-1) & \text{ if }c=0,\\
q^{s-1}(q^s+1) & \text{ if }c\in \gf_q^*.
\end{array}\right.
\end{eqnarray*}
\end{lemma}
\begin{proof}
By the orthogonal relation of additive characters and Lemma \ref{lem-sums},
\begin{eqnarray*}
n=\frac{1}{q}\sum_{x\in \gf_{q^m}^*}\sum_{y\in \gf_q}\zeta_{p}^{\tr_{q/p}\left(y\tr_{q^s/q}(x^{q^s+1})+yc\right)}=\frac{q^m-1}{q}+\frac{1}{q}\sum_{y\in \gf_q^*}\varphi(yc)\sum_{x\in \gf_{q^m}^*}\varphi\left(y\tr_{q^s/q}(x^{q^s+1})\right)=\frac{q^m-1}{q}+\frac{1}{q}\Delta(1,0,c).
\end{eqnarray*}
Then the desired conclusion follows from Lemma \ref{lem-sums}.
\end{proof}

Now we determine the value of the following exponential sums used to obtain the weight distribution of $\C_D$.
\begin{lemma}\label{value-scb}
Let $b\in \gf_{q^m}^*$ with $m=2s$. Denote by
$$S_c(b)=\sum_{x\in \gf_{q^m}^*}\sum_{y\in \gf_{q}^*}\sum_{z\in \gf_q^*}\varphi(zc)\varphi\left(y\tr_{q^m/q}(bx)+z\tr_{q^s/q}(x^{q^s+1})\right),\ c\in\gf_q.$$
If $c=0$ and $q$ is any prime power, then
\begin{eqnarray*}
S_0(b)&=&\begin{cases}
-(q-1)^2(q^s+1)& \text{ if $\tr_{q^s/q}(b^{q^s+1})=0$},\\
(q-1)(q^s-q+1)& \text{ if $\tr_{q^s/q}(b^{q^s+1})\neq 0$}.
\end{cases}
\end{eqnarray*}
If $c\in \gf_q^*$ and $q$ is even, then
\begin{eqnarray*}
S_c(b)&=&\begin{cases}
(q^s+1)(q-1)& \text{ if $\tr_{q^s/q}(b^{q^s+1})=0$},\\
q-1-q^s& \text{ if $\tr_{q^s/q}(b^{q^s+1})\neq 0$}.
\end{cases}
\end{eqnarray*}
If $c\in \gf_q^*$ and $q=p^e$ is odd with $\eta$ the quadratic character of $\gf_q^*$,
$$S_c(b)=\begin{cases}
(q^s+1)(q-1)& \text{ if $\tr_{q^s/q}(b^{q^s+1})=0$},\\
(q-1)-q^s\left(1+(-1)^{(\frac{p-1}{2})^2e}q\eta\left(-c\tr_{q^s/q}\left(b^{q^s+1}\right)\right)\right)
& \text{ if $\tr_{q^s/q}(b^{q^s+1})\neq 0$}.\\
\end{cases}$$
\end{lemma}

\begin{proof}
Recall that $m=2s$. By the transitivity of trace functions, we have
\begin{eqnarray*}
S_c(b)&=&\sum_{y\in \gf_{q}^*}\sum_{z\in \gf_q^*}\varphi(zc)\sum_{x\in \gf_{q^m}^*}\varphi\left(z\left(\tr_{q^s/q}(x^{q^s+1})+\tr_{q^m/q}\left(\frac{y}{z}bx\right)\right)\right)\\
&=&\sum_{y\in \gf_{q}^*}\sum_{z\in \gf_q^*}\varphi(zc)\sum_{x\in \gf_{q^m}^*}\varphi\left(z\left(\tr_{q^s/q}(x^{q^s+1})+\tr_{q^s/q}\left(\tr_{q^m/q^s}\left(\frac{y}{z}bx\right)\right)\right)\right)\\
&=& \sum_{y\in \gf_{q}^*}\sum_{z\in \gf_q^*}\varphi(zc)\sum_{x\in \gf_{q^m}^*}\varphi\left(z\left(\tr_{q^s/q}\left(x^{q^s+1}+\frac{ybx}{z}+\frac{yb^{q^s}x^{q^s}}{z}\right)\right)\right).
\end{eqnarray*}
Note that
\begin{eqnarray*}
x^{q^s+1}+\frac{ybx}{z}+\frac{yb^{q^s}x^{q^s}}{z}=\left(x+\frac{yb^{q^s}}{z}\right)^{q^s+1}-\frac{y^2b^{q^s+1}}{z^2}
\end{eqnarray*}
as $y,z\in \gf_q^*$. Hence
\begin{eqnarray*}
S_c(b)&=&\sum_{y\in \gf_{q}^*}\sum_{z\in \gf_q^*}\varphi(zc)\sum_{x\in \gf_{q^m}^*}\varphi\left(z\left(\tr_{q^s/q}\left(\left(x+\frac{yb^{q^s}}{z}\right)^{q^s+1}-\frac{y^2b^{q^s+1}}{z^2}\right)\right)\right)\\
&=&\sum_{y\in \gf_{q}^*}\sum_{z\in \gf_q^*}\varphi\left(zc-\tr_{q^s/q}\left(\frac{y^2b^{q^s+1}}{z}\right)\right)\sum_{x\in \gf_{q^m}^*}\varphi\left(z\tr_{q^s/q}\left(x+\frac{yb^{q^s}}{z}\right)^{q^s+1}\right)\\
&=&-\sum_{y\in \gf_{q}^*}\sum_{z\in \gf_q^*}\varphi\left(zc-\tr_{q^s/q}\left(\frac{y^2b^{q^s+1}}{z}\right)\right)\varphi\left(\tr_{q^s/q}\left(\frac{y^2b^{q^s+1}}{z}\right)\right)\\
& &+\sum_{y\in \gf_{q}^*}\sum_{z\in \gf_q^*}\varphi\left(zc-\tr_{q^s/q}\left(\frac{y^2b^{q^s+1}}{z}\right)\right)\sum_{x\in \gf_{q^m}}\varphi\left(z\tr_{q^s/q}\left(\left(x+\frac{yb^{q^s}}{z}\right)^{q^s+1}\right)\right)\\
&=&-\sum_{y\in \gf_{q}^*}\sum_{z\in \gf_q^*}\varphi(zc)+\sum_{y\in \gf_{q}^*}\sum_{z\in \gf_q^*}\varphi\left(zc-\tr_{q^s/q}\left(\frac{y^2b^{q^s+1}}{z}\right)\right)\sum_{x'\in \gf_{q^m}}\varphi\left(\tr_{q^s/q}(zx'^{q^s+1})\right)\\
&=&-\sum_{y\in \gf_{q}^*}\sum_{z\in \gf_q^*}\varphi(zc)+\sum_{y\in \gf_{q}^*}\sum_{z\in \gf_q^*}\varphi\left(zc-\tr_{q^s/q}\left(\frac{y^2b^{q^s+1}}{z}\right)\right)\sum_{x'\in \gf_{q^m}}\phi\left(zx'^{q^s+1})\right),
\end{eqnarray*}
where we used the substitution $x+\frac{yb^{q^s}}{z}\mapsto x'$ in the forth equality and $\phi$ denotes the canonical additive character of $\gf_{q^s}$. Since
\begin{eqnarray*}
\sum_{x'\in \gf_{q^m}}\phi\left(zx'^{q^s+1})\right)= 1+\sum_{x'\in \gf_{q^m}^*}\phi\left(zx'^{q^s+1})\right)=1+(q^s+1)\sum_{x'\in \gf_{q^s}^*}\phi\left(zx'\right)=-q^s
\end{eqnarray*}
for $z\in \gf_q^*$, we have
\begin{eqnarray}\label{eqn-Scb}
\nonumber S_c(b)&=&-\sum_{y\in \gf_{q}^*}\sum_{z\in \gf_q^*}\varphi(zc)-q^s\sum_{y\in \gf_{q}^*}\sum_{z\in \gf_q^*}\varphi\left(zc-\tr_{q^s/q}\left(\frac{y^2b^{q^s+1}}{z}\right)\right)\\
\nonumber &=&-\sum_{y\in \gf_{q}^*}\sum_{z\in \gf_q^*}\varphi(zc)-q^s\sum_{z\in \gf_q^*}\varphi(zc)\sum_{y\in \gf_{q}^*}\varphi\left(-\tr_{q^s/q}\left(\frac{b^{q^s+1}}{z}\right)y^2\right)\\
&=&-\sum_{y\in \gf_{q}^*}\sum_{z\in \gf_q^*}\varphi(zc)+q^s\sum_{z\in \gf_q^*}\varphi(zc)-q^s\sum_{z\in \gf_q^*}\varphi(zc)\sum_{y\in \gf_{q}}\varphi\left(-\tr_{q^s/q}\left(\frac{b^{q^s+1}}{z}\right)y^2\right).
\end{eqnarray}
Now we determine the value of $S_c(b)$ in the following two cases:
\begin{enumerate}
\item Let $q$ be even. By Lemma \ref{lem-charactersum-evenq},
$$\sum_{y\in \gf_{q}}\varphi\left(-\tr_{q^s/q}\left(\frac{b^{q^s+1}}{z}\right)y^2\right)=\begin{cases}
q & \text{ if $\tr_{q^s/q}(b^{q^s+1})=0$},\\
0 & \text{ if $\tr_{q^s/q}(b^{q^s+1})\neq 0$}.
\end{cases}$$
By Equation (\ref{eqn-Scb}), we then have
\begin{eqnarray*}
S_c(b)&=&\begin{cases}
-\sum\limits_{y\in \gf_{q}^*}\sum\limits_{z\in \gf_q^*}\varphi(zc)+q^s\sum\limits_{z\in \gf_q^*}\varphi(zc)-q^{s+1}\sum\limits_{z\in \gf_q^*}\varphi(zc)& \text{ if $\tr_{q^s/q}(b^{q^s+1})=0$},\\
-\sum\limits_{y\in \gf_{q}^*}\sum\limits_{z\in \gf_q^*}\varphi(zc)+q^s\sum\limits_{z\in \gf_q^*}\varphi(zc) & \text{ if $\tr_{q^s/q}(b^{q^s+1})\neq 0$}.
\end{cases}
\end{eqnarray*}
Since
$$\sum\limits_{z\in \gf_q^*}\varphi(zc)=\begin{cases}
q-1& \text{ if $c=0$,}\\
-1& \text{ if $c\neq 0$}.
\end{cases}$$ the value of $S_c(b)$ follows.
\item Let $q$ be odd and $\eta$ the quadratic character of $\gf_q^*$. If $c\in \gf_q^*$, by Lemmas \ref{lem-charactersum} and \ref{quadGuasssum},
\begin{eqnarray*}
& &\sum_{z\in \gf_q^*}\varphi(zc)\sum_{y\in \gf_{q}}\varphi\left(-\tr_{q^s/q}\left(\frac{b^{q^s+1}}{z}\right)y^2\right)\\
&=&\begin{cases}
q\sum_{z\in \gf_q^*}\varphi(zc)& \text{ if $\tr_{q^s/q}(b^{q^s+1})=0$}\\
G(\eta,\varphi)\sum_{z\in \gf_q^*}\varphi(zc)\eta\left(-\tr_{q^s/q}\left(\frac{b^{q^s+1}}{z}\right)\right)& \text{ if $\tr_{q^s/q}(b^{q^s+1})\neq 0$}\\
\end{cases}\\
&=&\begin{cases}
-q& \text{ if $\tr_{q^s/q}(b^{q^s+1})=0$}\\
G(\eta,\varphi)\eta\left(-c\tr_{q^s/q}\left(b^{q^s+1}\right)\right)\sum_{z\in \gf_q^*}\varphi(zc)\eta(cz)& \text{ if $\tr_{q^s/q}(b^{q^s+1})\neq 0$}\\
\end{cases}\\
&=&\begin{cases}
-q& \text{ if $\tr_{q^s/q}(b^{q^s+1})=0$}\\
G(\eta,\varphi)^2\eta\left(-c\tr_{q^s/q}\left(b^{q^s+1}\right)\right)& \text{ if $\tr_{q^s/q}(b^{q^s+1})\neq 0$}\\
\end{cases}\\
&=&\begin{cases}
-q& \text{ if $\tr_{q^s/q}(b^{q^s+1})=0$},\\
(-1)^{(\frac{p-1}{2})^2e}q\eta\left(-c\tr_{q^s/q}\left(b^{q^s+1}\right)\right)& \text{ if $\tr_{q^s/q}(b^{q^s+1})\neq 0$}.\\
\end{cases}\\
\end{eqnarray*}
By Equation (\ref{eqn-Scb}), we then obtain
$$S_c(b)=\begin{cases}
(q^s+1)(q-1)& \text{ if $\tr_{q^s/q}(b^{q^s+1})=0$},\\
(q-1)-q^s\left(1+(-1)^{(\frac{p-1}{2})^2e}q\eta\left(-c\tr_{q^s/q}\left(b^{q^s+1}\right)\right)\right)
& \text{ if $\tr_{q^s/q}(b^{q^s+1})\neq 0$}.\\
\end{cases}$$
If $c=0$, by Lemmas \ref{lem-charactersum} and \ref{quadGuasssum},
\begin{eqnarray*}
\sum_{z\in \gf_q^*}\varphi(zc)\sum_{y\in \gf_{q}}\varphi\left(-\tr_{q^s/q}\left(\frac{b^{q^s+1}}{z}\right)y^2\right)&=&\begin{cases}
q(q-1)& \text{ if $\tr_{q^s/q}(b^{q^s+1})=0$}\\
G(\eta,\varphi)\sum_{z\in \gf_q^*}\eta\left(-\tr_{q^s/q}\left(\frac{b^{q^s+1}}{z}\right)\right)& \text{ if $\tr_{q^s/q}(b^{q^s+1})\neq 0$}\\
\end{cases}\\
&=&\begin{cases}
q(q-1)& \text{ if $\tr_{q^s/q}(b^{q^s+1})=0$}\\
G(\eta,\varphi)\sum_{z\in \gf_q^*}\eta\left(-\tr_{q^s/q}\left(b^{q^s+1}\right)z\right)& \text{ if $\tr_{q^s/q}(b^{q^s+1})\neq 0$}\\
\end{cases}\\
&=&\begin{cases}
q(q-1)& \text{ if $\tr_{q^s/q}(b^{q^s+1})=0$},\\
0& \text{ if $\tr_{q^s/q}(b^{q^s+1})\neq 0$},\\
\end{cases}\\
\end{eqnarray*}
where the second equation holds due to $\eta(z)=\eta(1/z)$ and the third one holds due to the orthogonal relation of multiplicative characters. By Equation (\ref{eqn-Scb}), we then obtain
$$S_0(b)=\begin{cases}
-(q^s+1)(q-1)^2& \text{ if $\tr_{q^s/q}(b^{q^s+1})=0$},\\
(q-1)(q^s-q+1)
& \text{ if $\tr_{q^s/q}(b^{q^s+1})\neq 0$}.\\
\end{cases}$$
\end{enumerate}
The proof is completed.
\end{proof}

In the following, we study the weight enumerator of $\C_{D}$.
\begin{theorem}\label{th-1}
Let $s$ be a positive integer and $m=2s$. If $c=0,s>1$ and $q$ is any prime power, then $\C_{D}$ is a $q$-ary $[(q^s+1)(q^{s-1}-1),m,(q^{2s-2}-q^{s-1})(q-1)]$ linear code with weight enumerator
$$1+(q^s+1)(q^{s}-q^{s-1})z^{(q^{2s-2}-q^{s-1})(q-1)}+(q^s+1)(q^{s-1}-1)z^{q^{2s-2}(q-1)}.$$ If $c\in \gf_q^*$ and $q$ is even, then $\C_{D}$ is a $q$-ary $[q^{s-1}(q^s+1),m,q^{2s-2}(q-1)]$ linear code with weight enumerator
$$1+(q^s+1)(q^{s-1}-1)z^{q^{2s-2}(q-1)}+(q^s+1)(q^{s}-q^{s-1})z^{q^{2s-2}(q-1)+q^{s-1}}.$$  If $c\in \gf_q^*$ and $q$ is odd, then $\C_{D}$ is a $q$-ary $[q^{s-1}(q^s+1),m,q^{2s-2}(q-1)]$ linear code with weight enumerator
$$1+\frac{(q^s+1)(q^s+q^{s-1}-2)}{2}z^{q^{2s-2}(q-1)}+\frac{q^{s-1}(q-1)(q^s+1)}{2}z^{q^{2s-2}(q-1)+2q^{s-1}}.$$
\end{theorem}
\begin{proof}
 For any codeword $$\bc(b)=\left(\tr_{q^m/q}(bd_1),\tr_{q^m/q}(bd_2),\cdots,\tr_{q^m/q}(bd_n)\right)\in \C_D$$ with $b\in \gf_{q^m}^*$, by the orthogonal relation of additive characters, its Hamming weight
 \begin{eqnarray}\label{eqn-1}
 \nonumber \wt(\bc(b))&=&n-\sharp \left\{x\in \gf_{q^m}^*: \tr_{q^m/q}(bx)=0\mbox{ and } \tr_{q^s/q}(x^{q^s+1})+c=0\right\}\\
 \nonumber &=&n-\frac{1}{q^2}\sum_{x\in \gf_{q^m}^*}\sum_{y\in \gf_{q}}\zeta_{p}^{y\tr_{q^m/q}(bx)}\sum_{z\in \gf_q}\zeta_{p}^{z\left(\tr_{q^s/q}(x^{q^s+1})+c\right)}\\
 \nonumber &=&n-\frac{1}{q^2}\sum_{x\in \gf_{q^m}^*}\sum_{y\in \gf_{q}}\sum_{z\in \gf_q}\varphi(zc)\varphi\left(y\tr_{q^m/q}(bx)+z\tr_{q^s/q}(x^{q^s+1})\right)\\
 \nonumber &=&n-\frac{q^m-1}{q^2}-\frac{1}{q^2}\sum_{y\in \gf_{q}^*}\sum_{x\in \gf_{q^m}^*}\varphi\left(y\tr_{q^m/q}(bx)\right)-\frac{1}{q^2}\sum_{z\in \gf_{q}^*}\sum_{x\in \gf_{q^m}^*}\varphi(zc)\varphi\left(z\tr_{q^s/q}(x^{q^s+1})\right)\\
\nonumber & &-\frac{1}{q^2}\sum_{x\in \gf_{q^m}^*}\sum_{y\in \gf_{q}^*}\sum_{z\in \gf_q^*}\varphi(zc)\varphi\left(y\tr_{q^m/q}(bx)+z\tr_{q^s/q}(x^{q^s+1})\right)\\
\nonumber &=&n-\frac{q^m-1}{q^2}-\frac{1}{q^2}\sum_{y\in \gf_{q}^*}\sum_{x\in \gf_{q^m}^*}\varphi\left(y\tr_{q^m/q}(bx)\right)-\frac{1}{q^2}\sum_{z\in \gf_{q}^*}\sum_{x\in \gf_{q^m}^*}\varphi(zc)\varphi\left(z\tr_{q^s/q}(x^{q^s+1})\right)-\frac{1}{q^2}S_c(b),
 \end{eqnarray}
 where $S_c(b)$ was defined in Lemma \ref{value-scb}.
 Let $\chi$ be the canonical additive character of $\gf_{q^m}$. By the transitivity of trace functions,
 \begin{eqnarray}\label{eqn-2}
 \sum_{y\in \gf_{q}^*}\sum_{x\in \gf_{q^m}^*}\varphi\left(y\tr_{q^m/q}(bx)\right)= \sum_{y\in \gf_{q}^*}\sum_{x\in \gf_{q^m}^*}\chi(ybx)=-(q-1).
 \end{eqnarray}
By Lemma \ref{lem-sums},
\begin{eqnarray}\label{eqn-3}
 \nonumber \sum_{z\in \gf_{q}^*}\sum_{x\in \gf_{q^m}^*}\varphi(zc)\varphi\left(z\tr_{q^s/q}(x^{q^s+1})\right)=\Delta(1,0,c)=\left\{\begin{array}{ll}
(q^s+1)(1-q) & \text{ if }c=0,\\
q^s+1 & \text{ if }c\in \gf_q^*.
\end{array}\right.
 \end{eqnarray}
 Now we discuss the Hamming weight in several cases as follows.
 \begin{enumerate}
 \item\label{case-1} Let $c=0$. By Equations (\ref{eqn-1}), (\ref{eqn-2}), (\ref{eqn-3}) and Lemmas \ref{length} and \ref{value-scb}, we directly have
 \begin{eqnarray*}
 \wt(\bc(b))&=&\begin{cases}
q^{2s-2}(q-1)& \text{ if $\tr_{q^s/q}(b^{q^s+1})=0$},\\
(q^{2s-2}-q^{s-1})(q-1)& \text{ if $\tr_{q^s/q}(b^{q^s+1})\neq 0$}.
\end{cases}
 \end{eqnarray*}
 For $b\in \gf_{q^m}^*$, $\tr_{q^s/q}\left(b^{q^s+1}\right)=0\Leftrightarrow b^{q^s+1}\in \ker(\tr_{q^s/q})\backslash\{0\}.$
For fixed $b'\in \ker(\tr_{q/p})\backslash\{0\}$, there exist $q^s+1$ elements $b$ such that $b^{q^s+1}=b'$. Hence
$A_{q^{2s-2}(q-1)} =(q^s+1)(q^{s-1}-1)$ and
$$A_{(q^{2s-2}-q^{s-1})(q-1)}=q^{2s}-1-A_{q^{2s-2}(q-1)}=(q^s+1)(q^{s}-q^{s-1}).$$
Then the weight distribution follows in this case.
\item Let $c\in \gf_q^*$ and $q$ be even.By Equations (\ref{eqn-1}), (\ref{eqn-2}), (\ref{eqn-3}) and Lemmas \ref{length} and \ref{value-scb}, we directly have
 \begin{eqnarray*}
 \wt(\bc(b))&=&\begin{cases}
q^{2s-2}(q-1)& \text{ if $\tr_{q^s/q}(b^{q^s+1})=0$},\\
q^{2s-2}(q-1)+q^{s-1}& \text{ if $\tr_{q^s/q}(b^{q^s+1})\neq 0$}.
\end{cases}
\end{eqnarray*}
Similarly to Case \ref{case-1}), we have $A_{q^{2s-2}(q-1)} =(q^s+1)(q^{s-1}-1)$ and $A_{q^{2s-2}(q-1)+q^{s-1}}=(q^s+1)(q^{s}-q^{s-1})$. Then the weight distribution follows in this case.
\item Let $c\in \gf_q^*$ and $q=p^e$ be odd. By Equations (\ref{eqn-1}), (\ref{eqn-2}), (\ref{eqn-3}) and Lemmas \ref{length} and \ref{value-scb}, we directly derive
$$\wt(\bc(b))=\begin{cases}
q^{2s-2}(q-1)& \text{ if $\tr_{q^s/q}(b^{q^s+1})=0$},\\
q^{2s-2}(q-1)\text{ or }q^{2s-2}(q-1)+2q^{s-1}
& \text{ if $\tr_{q^s/q}(b^{q^s+1})\neq 0$}.\\
\end{cases}$$
The dimension of $\C_D$ is $m$ as $A_0=1$. Since $0\in D$, the minimal distance $d^{\perp}$ of the dual of $\C_D$ can't be 1. Then $d^{\perp}\geq 2$.
By the first two Pless power moments, we have
\begin{eqnarray*}\begin{cases}
A_{q^{2s-2}(q-1)}+A_{q^{2s-2}(q-1)+2q^{s-1}}=q^{m}-1, \\
q^{2s-2}(q-1)A_{q^{2s-2}(q-1)}+(q^{2s-2}(q-1)+2q^{s-1})A_{q^{2s-2}(q-1)+2q^{s-1}}=q^{m-1}(qn-n).
\end{cases}
\end{eqnarray*}
Solving this equation system yields $A_{q^{2s-2}(q-1)}=\frac{(q^s+1)(q^s+q^{s-1}-2)}{2}$ and
$A_{q^{2s-2}(q-1)+2q^{s-1}}=\frac{q^{s-1}(q-1)(q^s+1)}{2}$. Then the weight distribution follows in this case.
 \end{enumerate}
 The proof is completed.
\end{proof}

Examples of the linear code $\C_D$ are summarized in Table \ref{tab-1}. Some of them are optimal or almost optimal and some of them have the same parameters as the best linear codes according to the Code Tables at https://www.codetables.de/.
\begin{table}[ht]
\begin{center}
\caption{Examples of $\C_D$ in Theorem \ref{th-1}}\label{tab-1}
\begin{tabular}{ccccccc} \hline
Length &  Dimension & Minimal distance&  $c$& $m$ & $q$ & Optimality   \\ \hline
5 & 4 & 2 & 0 & 4 & 2 & Optimal\\
27 & 6 & 12 &0 & 6 & 2 & Optimal\\
51 & 4 & 36 & 0 & 4 & 4 & Best known\\
5 & 2 & 4 & 1 & 2 & 4 & Optimal\\
64 & 4 & 48 & 1 & 4 & 4 & Optimal\\
9 & 2 & 8 & 1 & 2 & 8 & Optimal\\
10 & 2 & 8 & 1 & 2 & 9 & Almost Optimal\\
30 & 4 & 18 & 1 & 4 & 3 & Almost optimal\\
\hline
\end{tabular}
\end{center}
\end{table}

Denote by $(1,A_1^{\perp},A_2^{\perp},\cdots,A_n^{\perp})$ the weight distribution of $\C_{D}^{\perp}$. There exists a relationship called Pless power moments between $(1,A_1,A_2,\cdots,A_n)$ and $(1,A_1^{\perp},A_2^{\perp},\cdots,A_n^{\perp})$ \cite[Page 259]{HP}. Pless power moments are useful for determining the minimal distance of $\C_{D}^{\perp}$.
\begin{theorem}\label{th-2}
Let $m=2s$. If $q=2,s\geq 3$ and $c=0$, then $\C_{D}^{\perp}$ is a $[(q^s+1)(q^{s-1}-1),(q^s+1)(q^{s-1}-1)-m,3]$ linear code. If $q>2,s>1$ and $c=0$, then $\C_{D}^{\perp}$ is a $[(q^s+1)(q^{s-1}-1),(q^s+1)(q^{s-1}-1)-m,2]$ linear code. If $q$ is even and $c\in \gf_q^*$, then $\C_{D}^{\perp}$ is a $[q^{s-1}(q^s+1),q^{s-1}(q^s+1)-m,3]$ linear code. If $q$ is odd and $c\in \gf_q^*$, then $\C_{D}^{\perp}$ is a $[q^{s-1}(q^s+1),q^{s-1}(q^s+1)-m,2]$ linear code.
\end{theorem}
\begin{proof}
Since $0\not\in D$, the minimal distance $d^{\perp}$ of the dual satisfies $d^{\perp}\geq 2$. By the second Pless power moment,
\begin{eqnarray}\label{eqn-pless-1}
\sum_{j=0}^{n}j^2A_j=q^{m-2}(q-1)n(qn-n+1)- 2q^{m-2}A_2^\perp.\
\end{eqnarray}
\begin{enumerate}
\item If $c=0$, by Theorem \ref{th-1} and Equation (\ref{eqn-pless-1}),
$A_2^{\perp}=\frac{(q-1)(q-2)(q^s+1)(q^{s-1}-1)}{2}$.
Hence $A_2^{\perp}=0$ if and only if $q=2$ as $s>1$. Let $q=2$, by the third Pless power moment,
\begin{eqnarray}\label{eqn-pless-2}
\sum_{j=0}^{n}j^3A_j=2^{2s-3}[n^3(n+3)-6A_3^{\perp}].
\end{eqnarray}
By Theorem \ref{th-1} and Equation (\ref{eqn-pless-2}),
$A_3^{\perp}=\frac{2^{3s-4}(2^s-3)-2^{s-2}(2^{s+1}-3)+1}{3}>0$ if $s\geq 3$ and $q=2$. Hence $d^{\perp}=3$ if $q=2,s\geq 3$ and $d^{\perp}=2$ if $q>2,s>1$.
\item If $c\in \gf_q^*$ and $q$ is even, by Theorem \ref{th-1} and Equation (\ref{eqn-pless-1}), we have $A_2^{\perp}=0$. Then $d^{\perp}\geq 3$. By the third Pless power moment,
    \begin{eqnarray}\label{eqn-pless-3}
    \sum_{j=0}^{n}j^3A_j=q^{2s-3}[(q-1)n(q^2n^2-2qn^2+3qn-q+n^2-3n+2)-6A_3^\perp].
    \end{eqnarray}
By Theorem \ref{th-1} and Equation (\ref{eqn-pless-3}),
$A_3^{\perp}=\frac{q^{s-1}(q^s+1)(q-1)\left(q^{s-1}(2q-3)+q^{2s-2}(q-1)^2+2-q\right)}{6}>0$ for $c\in \gf_q^*$ and even $q$. Then $d^{\perp}=3$.
\item If $c\in \gf_q^*$ and $q$ is odd, by Theorem \ref{th-1} and Equation (\ref{eqn-pless-1}), we have
$A_2^{\perp}=\frac{q^{s-1}(q-1)(q^s+1)}{2}>0$. Then $d^{\perp}=2$.
\end{enumerate}
The proof is completed.
\end{proof}
\section{Strongly regular graphs with new parameters}\label{section}
A connected graph of $N$ vertices is called \emph{strongly regular} with parameters $(N,K,\lambda,\mu)$ if it is regular with valency $K$ and if the number of vertices joined to two given vertices is $\lambda$ or $\mu$ according as the two given vertices are adjacent or non-adjacent.

 Let $G=[\textbf{y}_1,\textbf{y}_2,\cdots,\textbf{y}_n]$ be the generator matrix of a $[n,k]$ linear code $\C$ over $\gf_q$, where $\textbf{y}_i\in \gf_{q}^k$. Let $\textbf{V}=\gf_q^k$, $\textbf{O}=\{\langle\textbf{y}_i\rangle:i=1,2,\cdots,n\}$ and
$\Omega=\{\textbf{v}\in\textbf{V}:\langle \textbf{v}\rangle\in \textbf{O} \}$. Define a graph $G(\Omega)$  with vertices the vectors of $\textbf{V}$ and where two vertices are joint if and only if their difference is in $\Omega$. In \cite[Th. 3.2]{C}, it was proved that $G(\Omega)$ is strongly regular if and only if $\C$ is a projective two-weight code. Assume that the nonzero weights of $\C$ are $w_1$ and $w_2$. By \cite[Coro. 3.7]{C}, the parameters of $G(\Omega)$ are
$$N=q^k,\ K=n(q-1),\ \lambda=K^2+3K-q(w_1+w_2)-Kq(w_1+w_2)+q^2w_1w_2,\ \mu=K^2+K-Kq(w_1+w_2)+q^2w_1w_2.$$
By Theorems \ref{th-1} and \ref{th-2}, if $m=2s$, $q$ is even and $c\in \gf_q^*$, then $\C_D$ is a projective $[q^{s-1}(q^s+1),m,q^{2s-2}(q-1)]$ two-weight linear code with weight enumerator
$$1+(q^s+1)(q^{s-1}-1)z^{q^{2s-2}(q-1)}+(q^s+1)(q^{s}-q^{s-1})z^{q^{2s-2}(q-1)+q^{s-1}}.$$
Hence $\C_D$ yields a strongly regular graph $G(\Omega)$ with the following parameters:
\begin{eqnarray*}
N=q^{2s},\ K=q^{s-1}(q^s+1)(q-1),\ \lambda=q^{s-1}\left(2q-3+q^{s-1}(q-1)^2\right),\ \mu=(q-1)q^{s-1}(q^s-q^{s-1}+1).
\end{eqnarray*}
Compared with known  strongly regular graphs \cite{C}, $G(\Omega)$ has new parameters.

\section{Summary and concluding remarks}
The contributions of this paper are summarized as follows:
\begin{enumerate}
\item[$\bullet$] If $m=2s$, the weight distribution of $\C_D$ was determined in Theorem \ref{th-1} for any prime power $q$ and any $c\in \gf_q$;
\item[$\bullet$] The parameters of the dual of $\C_D$ was also determined in Theorem \ref{th-2};
\item[$\bullet$]  By Theorems \ref{th-1} and \ref{th-2}, $\C_D$ is a projective two-weight code if $q$ is even and $c\in \gf_q^*$. This class of projective two-weight codes yield strongly regular graphs with new parameters in Section \ref{section}.
\end{enumerate}
It is easy to prove that $\C_D$ in Theorem \ref{th-1} is minimal for $s\geq 3$ as $w_{\min}/w_{\max}>\frac{q-1}{q}$ \cite[Lemma 3]{YD}, where $w_{\min},w_{\max}$ denote the maximum and minimum weights of $\C_D$, respectively. Hence the code $\C_D$ in Theorem \ref{th-1} can also be used to construct secret sharing schemes with interesting access structures \cite{YD}. As another application, the projective two-weight code $\C_D$ for even $q$ and $c\in \gf_q^*$ holds $1$-design by the well-known Assmus-Mattson theorem \cite{Dingbook18}.

\end{document}